\documentclass[article]{article}

\usepackage{a4wide}
\usepackage{algorithm}
\usepackage{algpseudocode}
\usepackage{amsfonts}
\usepackage{amsmath}
\usepackage{amssymb}
\usepackage{amsthm}
\usepackage{array}
\usepackage[backend=biber, style=nature, sorting=none, intitle=true]{biblatex}
\usepackage{bbm}
\usepackage{bm} 
\usepackage{color}
\usepackage{comment}
\usepackage{endnotes}
\usepackage{enumerate}
\usepackage{epigraph}
\usepackage{epstopdf}
\usepackage[draft]{fixme}
\usepackage[T1]{fontenc}
\usepackage{fullpage}
\usepackage{graphicx}
\usepackage[bookmarks, bookmarksdepth=2, colorlinks=true, linkcolor=blue, citecolor=blue, urlcolor=blue]{hyperref}
\usepackage{cleveref} 
\usepackage{ifpdf}
\usepackage{lmodern}
\usepackage{longtable}
\usepackage{makecell}
\usepackage{mathrsfs}
\usepackage{mathtools}
\usepackage{multicol}
\usepackage{multirow}
\usepackage{parskip}
\usepackage{subfig}
\usepackage{url}
\usepackage{wrapfig}
\usepackage{xparse}

\usepackage{tikz}
\usepackage{tikz-cd}
\usetikzlibrary{positioning}

\newtheorem{theorem}{Theorem}
\newtheorem{definition}{Definition}
\newtheorem{lemma}[theorem]{Lemma}
\newtheorem{remark}{Remark}
\newtheorem{example}{Example}

\DeclareMathOperator{\dist}{\delta}
\DeclareMathOperator{\Dr}{Dr} 
\DeclareMathOperator{\Gr}{Gr} 
\DeclareMathOperator{\tropGr}{tropGr} 
\DeclareMathOperator{\lin}{lin} 

\DeclareMathOperator{\valop}{val} 

\newcommand{\ones}{\mathbf{1}} 
\newcommand{\bergFan}{\widetilde{\mathcal{B}}} 
\newcommand \dvec{{\mathrm d}} 
\newcommand \Dmat{{\mathrm D}} 

\newcommand{\dtr}{d_{\mathrm{tr}}}
\newcommand{\edgeset}{E} 
\newcommand \K{{\mathbb K}} 
\newcommand \numleaves{\ensuremath{n}} 
\newcommand \numpairs{\ensuremath{m}} 
\newcommand{\NAAV}[1]{\lvert #1 \rvert_{\nu}} 
\newcommand{\PAAV}[1]{\lvert #1 \rvert_{p}} 
\newcommand \plucker{{\rho}} 
\newcommand \Q{{\mathbb Q}}
\newcommand \R{{\mathbb R}}
\newcommand \Z{{\mathbb Z}}
\newcommand \shift{{\epsilon}} 
\newcommand \treeroot{{\omega}} 
\newcommand{\Treespace}{\mathcal{T}} 
\newcommand{\Tmax}{\mathbb{T}_{\max}}

\newcommand{\TPS}[1]{\mathbb{TP}^{#1}} 
\newcommand{\TPT}[1]{\mathbb{R}^{#1}/\mathbb{R}\mathbf{1}} 

\newcommand{\Uspace}{\mathcal{U}} 
\NewDocumentCommand{\val}{o}{%
    \ensuremath{\IfNoValueTF{#1}{\mathrm{val}}{\valop\left(#1\right)}}}
\newcommand \vecOne{\ensuremath{u}} 
\newcommand \vecTwo{\ensuremath{v}} 
\newcommand \x{{\mathbf{x}}}
\newcommand{\RightComment}[1]{\hfill{\footnotesize\textcolor{gray}{#1}}}

\title{Phylogenetics in a warm place: computational aspects of the Tropical Grassmannian}

\author{%
Samir Bhatt\textsuperscript{1,2,*},
John Sabol\textsuperscript{3,*},
Papri Dey\textsuperscript{4,*},
Matthew J. Penn\textsuperscript{1,*},\\
David Duchene\textsuperscript{1},
Ruriko Yoshida\textsuperscript{3}
}

\date{}

\begin{document}
\maketitle

\begingroup
\renewcommand{\thefootnote}{}
\footnotetext{%
\textsuperscript{1} University of Copenhagen;
\textsuperscript{2} Imperial College London;
\textsuperscript{3} Naval Postgraduate School;
\textsuperscript{4} University of California, Santa Cruz.
 \textsuperscript{*} These authors contributed equally to this work.
}
\endgroup
\setcounter{footnote}{0}

\begin{abstract}
    Phylogenetic trees provide a fundamental representation of evolutionary relationships, yet the combinatorial explosion of possible tree topologies renders inference computationally challenging. Classical approaches to characterizing tree space, such as the Billera–Holmes–Vogtmann (BHV) space, offer elegant geometric structure but suffer from statistical and computational limitations. An alternative perspective arises from tropical geometry: the tropical Grassmannian ($\tropGr(2,n)$), introduced by Speyer and Sturmfels, which coincides with phylogenetic tree space. In this paper, we review the structure of the tropical Grassmannian and present algorithmic methods for its computational study, including procedures for sampling from the tropical Grassmannian. Our aim is to make these concepts accessible to evolutionary biologists and computational scientists, and to motivate new research directions at the interface of algebraic geometry and phylogenetic inference.\footnote{Code for all figures are available here \url{https://github.com/bhattsamir/tropical_grassmannian}}
\end{abstract}

A phylogenetic tree is a mathematical structure that connects the worlds of evolutionary biology and computer science. In evolutionary biology, it represents the evolutionary relationships among a set of taxa. In computer science, an unrooted phylogenetic tree on a set of $\numleaves$ taxa is defined as a tree $T\in \Treespace_\numleaves = (\mathcal{V}, \mathcal{E})$, where $\mathcal{V}$ is the set of vertices and $\mathcal{E}$ the set of edges (often weighted). The space of possible trees $\Treespace_\numleaves$ is finite, but intractably large and given by the Schröder number $(2\numleaves-5)!!$ in the unrooted case. The tree $T$ has exactly $\numleaves$ leaves (also called tips), each corresponding to one taxon, and $\numleaves - 2$ internal nodes, each of degree 3. For just $\numleaves=40$ individuals or species, the number of possible evolutionary trees is already greater than the number of hydrogen atoms in the Sun. Trying to find the correct tree in that space is like trying to locate one specific hydrogen atom somewhere inside the entire Sun. If the tree is rooted, the root node has out-degree 2, and there are $\numleaves - 1$ internal nodes. In both rooted and unrooted cases, the edges $\mathcal{E}$ are typically assigned non-negative weights, representing evolutionary distances or divergence times.

The central challenge in computational phylogenetics is finding the best tree, $T^*$ given some data $X$ and an objective criterion $\mathcal{L}$ that measures how well $T^*$ characterises the data.
i.e. 
\begin{align}
    T^* = \arg\min_{T \in \Treespace} \; \mathcal{L}(X, T).
    \label{eq:main_optimisation_problem}
\end{align}
Broadly three major objective criteria exist, maximum parsimony \cite{Fitch1971-tz}, maximum likelihood \cite{Felsenstein1983-wd} and minimum evolution \cite{Rzhetsky1993-vo}. Solving Equation \ref{eq:main_optimisation_problem} is NP-hard for all three criteria \cite{Foulds1982-wu,Roch2006-qs,Fiorini2012-lq}. This means that there is no known polynomial-time algorithm that always finds the globally optimal phylogenetic tree under standard objective criteria. Heuristic algorithms, such as hill climbing and tree rearrangement strategies often perform well in practice, but offer no guarantee of global optimality. Hill climbing performs well in phylogenetic tree search because the objective criteria landscapes on tree space exhibit strong local correlation, producing large basins of attraction around high-likelihood topologies and enabling local tree change moves to reliably ascend toward near-optimal trees. Therefore a central challenge in these heuristic algorithms is efficiently exploring tree space. This exploration is typically performed through local rearrangement operations that generate neighboring trees. Two common operations are Nearest Neighbor Interchange (NNI) and Subtree-Prune and Regraft (SPR). NNI modifies a tree by swapping subtrees across an internal edge, resulting in a new topology with minimal change. SPR involves cutting a subtree from the tree and reattaching it at a different edge, allowing for broader exploration of tree space. These operations define the neighborhood structure of the search, but naturally raise the question: what exactly is tree space, and how is it structured? A geometrically rigorous solution to this question was provided by Billera, Holmes, and Vogtmann \cite{Billera2001-il}, who showed that phylogenetic tree space can be endowed with the structure of a CAT$(0)$ (Cartan–Alexandrov–Toponogov) space, a type of non-positively curved metric space that admits unique geodesics between points. The Billera--Holmes--Vogtmann (BHV) tree space provides a geometric and combinatorial model for the space of phylogenetic trees with edge lengths. Formally, it is a piecewise Euclidean cubical complex in which each orthant (i.e., a Euclidean space $\R_{\geq 0}^k$) corresponds to a unique tree topology defined by a fixed set of internal splits. The dimension $k$ of an orthant equals the number of internal edges in the tree, which is $\numleaves - 3$ for a fully resolved unrooted binary tree with $\numleaves$ leaves. Within each orthant, trees differ only in the lengths of their internal edges and are equipped with standard Euclidean geometry. Orthants are glued together along shared lower-dimensional faces corresponding to unresolved trees (i.e., trees with collapsed edges or polytomies). The resulting space is connected, contractible, and forms a CAT(0) space. This structure enables well-defined and efficiently computable geodesic distances between trees in polynomial time ($\mathcal{O}(\numleaves^4)$) \cite{Owen2011-sm}.

Despite its elegant geometric properties, BHV tree space presents certain challenges, particularly for statistical inference. One such issue is \emph{stickiness}, a phenomenon arising from the non-manifold boundaries between orthants. Because the space is composed of orthants glued along lower-dimensional faces, geodesics and Fréchet means (i.e., the average tree) often lie on these lower-dimensional boundaries, corresponding to unresolved trees with one or more zero-length edges. As a result, even when input trees are fully resolved, their average under the BHV metric can be topologically unresolved. Additionally, the high-dimensional and combinatorial structure of BHV space makes some computational tasks, such as likelihood optimization or Bayesian posterior integration, difficult to implement efficiently, particularly as the number of taxa increases. In addition, Lin et al. \cite{convexity} showed that the dimension of the convex hull of three points in terms of the BHV metric over BHV space is unbounded in general, and that a convex hull in terms of the BHV metric over BHV space might not be closed. Therefore, unlike Euclidean space, it is not trivial to conduct a simple statistical analysis over BHV space. Finally, it is challenging to efficiently sample uniformly from the BHV space, again prohibiting Bayesian inference. These challenges have motivated the exploration of alternative models of tree space that better accommodate statistical and algorithmic needs.

A completely different way to fully characterise tree space that is still geometrically rigorous was provided by Speyer and Sturmfels \cite{Speyer2004-bg} and called the \emph{Tropical Grassmannian}. Before explaining what this exotically named object is, we first need to introduce the concept of the tree metric. A metric space $(X, \dist)$ (where $X$ is the set of taxa of interest and $\dist$ is some distance function) is called a \textit{tree metric} if there exists a weighted tree $T$ such that for all $a, b \in X$, $\dist(a, b)$ equals the sum of edge weights along the unique path between $a$ and $b$ in $T$. Stated more simply, the leaf to leaf distance between any two taxa, is the sum of the branch lengths between them. Unfortunately, it is exceedingly unlikely that any estimated evolutionary distance matrices will meet this constraint. Therefore, one might ask, ``is there a sufficient condition for a matrix to be tree metric?'' Such a condition would then allow practitioners to know if a given distance matrix is valid. Given a tree, the cophenetic vector, which records the pairwise distances between taxa as measured by the height of their least common ancestor in the tree, naturally defines a tree metric. However, for an arbitrary distance matrix (also called a dissimilarity map) to qualify as a metric, it must satisfy the following properties:
\[
\begin{aligned}
&\text{(i) } \dist(i,i) = 0 \quad \forall i \\
&\text{(ii) } \dist(i,j) = \dist(j,i) \quad \forall i,j \\
&\text{(iii) } \dist(i,j) \ge 0 \quad \forall i,j \\
&\text{(iv) } \dist(i,j) \le \dist(i,k) + \dist(k,j) \quad \forall i,j,k.
\end{aligned}
\]
Further, for this distance matrix to be a tree metric, Buneman \cite{Buneman1974-np} introduced a necessary and sufficient condition called the \textit{Four-Point Condition}.
\begin{theorem}[Four-Point Condition]
  \label{thm:four_point_condition}
  Let $(X,\dist)$ be a metric space, and let $x_1,x_2,x_3,x_4 \in X$. Then
  \[
    \dist(x_1,x_2) + \dist(x_3,x_4)
    \;\le\;
    \max \bigl\{\, \dist(x_1,x_3) + \dist(x_2,x_4),\;\; \dist(x_1,x_4) + \dist(x_2,x_3) \bigr\}.
  \]
\end{theorem}

The theorem can be stated equivalently as the two largest sums among the following are equal:
\begin{equation}
    \dist(x_1, x_2) + \dist(x_3, x_4), \quad
    \dist(x_1, x_3) + \dist(x_2, x_4), \quad
    \dist(x_1, x_4) + \dist(x_2, x_3).
\end{equation}
A proof of this theorem is available in most standard phylogenetic textbooks. However, it can be understood without a formal proof: in a tree, there is a unique path between any pair of nodes. When considering four leaves, the distances between them are constrained by the tree’s branching structure. Among the three possible ways to pair up the leaves into two disjoint pairs, the two pairs that share the longest common path through the tree will have the largest combined distances. The four-point condition captures this by requiring that the largest two of the three pairwise sums are equal. This reflects the fact that in a tree, any four points must fit into a consistent subtree structure, typically a quartet tree and the four-point condition detects whether such a structure exists.

The four-point condition characterizes tree metrics, meaning that if a distance matrix satisfies this condition, it is a tree metric and therefore represents a valid phylogenetic tree. Such a tree can then be efficiently reconstructed using algorithms such as neighbour joining. A sensible question would therefore be, as opposed to enumerating tree space by constructing discrete trees \cite{Penn2024-ja}, is there a way to define the space of all possible trees through the four-point condition? This is exactly what the Tropical Grassmannian achieves. In fact, the space of all four-point conditions was developed before the seminal work of Speyer and Sturmfels by Dress \cite{Dress1992-zm, Herrmann2009-ad} under the name \emph{rank 2 valuated matroids}, a space we now call the \emph{Dressian} $\Dr(2,n)$, which is just the Dressian $\Dr(r,n)$ for the specific case of $r=2$. As we shall see, the term ``valuated matroid'' correctly describes the space of phylogenetic trees. 

\subsection*{The Grassmannian $\Gr(2,n)$}

The Grassmannian $\Gr(r,n)$ is a widely used differentiable manifold. Here, we focus on the specific Grassmannian for $r=2$, i.e., $\Gr(2,n)$, which is the space of 2-dimensional linear subspaces of an $n$-dimensional vector space defined over a field $\mathbb{K}$. In certain applications, it is desirable to use an algebraically closed field, such as the Puiseux series (a power series that allows fractional exponents,
forming an algebraically closed field $\mathbb{C}\{\{t\}\}$ that ensures every
polynomial equation has a solution). For ease of exposition here, however, we will identify $\mathbb{K}$ with the reals. Having done so, we can formally define:
\begin{equation}
    \Gr(2,n) = \left\{ V \subset \R^n \mid \dim(V) = 2 \right\}.
\end{equation}
Equivalently, we can define the Grassmannian as the span of all linearly independent pairs of \emph{column} vectors:
\begin{equation}
    \Gr(2,n) =
    \left\{ \operatorname{span}(\vecOne, \vecTwo) \mid \vecOne, \vecTwo \in \R^n,\, \lambda_1,\lambda_2 \in \R,\, 
    \lambda_1 \vecOne + \lambda_2 \vecTwo = 0 \iff \lambda_1 = \lambda_2 = 0
    \right\}.
\end{equation}

\subsection*{Pl\"ucker Coordinates and the Pl\"ucker Ideal}

As we have noted, the Grassmannian $\Gr(2,n)$ is the space of all 2-dimensional subspaces of $\R^n$ (or $\mathbb{K}^n$). To describe such subspaces, one can use a basis matrix ($n \times 2$ matrix), as in $V := \begin{pmatrix} u & v \end{pmatrix}$, but this representation clearly isn't unique since for any $\lambda_1,\lambda_2\neq 0$, $V':= \begin{pmatrix}\lambda_1 \vecOne & \lambda_2 \vecTwo \end{pmatrix}$ span the exact same subspace. Pl\"ucker coordinates provide a more natural, invariant way to describe elements of $\Gr(2,n)$.

The Pl\"ucker coordinates of a matrix $V$ are given by the determinants of all maximal minors (the $2 \times 2$ submatrices) of $V$. Thus, given an ordering on the rows of $V$ (to fix our indices), we can define the \emph{Pl\"ucker vector} ($\plucker$) as the image of a map that sends our matrix $V\in\R^{n\times2}$ to its unique embedding in projective space $\mathbb{P}^{(\numpairs-1)}$:
\begin{equation}\label{eq:plucker_computation}
    V = \begin{pmatrix} \vecOne_1 & \vecTwo_1 \\ \vecOne_2 & \vecTwo_2 \\  \vdots & \vdots \\ \vecOne_n & \vecTwo_n
    \end{pmatrix},
    \quad
   \plucker=\left(\,\plucker_\edgeset\right),\, \edgeset := \{ (i,j) \mid i<j \in [n]\} 
    \quad \textrm{and}
    \quad
    \plucker_{ij} := \det \begin{pmatrix} \vecOne_i & \vecTwo_i \\ \vecOne_j & \vecTwo_j \end{pmatrix} = \vecOne_i \vecTwo_j - \vecOne_j \vecTwo_i,
\end{equation}
where $\numpairs=|\edgeset|=\binom{n}{2}$. The notation $\edgeset$ here denotes all the pairs of vertices of $T=(\mathcal{V},\mathcal{E})$ such that both vertices are leaves in $T$. Critically, all these Pl\"ucker coordinates \emph{also} satisfy a series of quadratic Pl\"ucker relations given by the set of $\binom{n}{4}$ quadrics of the form:
\begin{equation}\label{eq:plucker_relations}
  \plucker_{ij} \plucker_{kl} - \plucker_{ik} \plucker_{jl} + \plucker_{il} \plucker_{jk} = 0, \quad \text{for all } 1 \leq i < j < k < l \leq n.  
\end{equation}
These quadratic Pl\"ucker relations are homogeneous in the coordinates $\plucker_\edgeset$, which justifies viewing the Pl\"ucker vector projectively to ensure consistency under scaling. The set of all such quadratic relations generates the Pl\"ucker ideal $\mathcal{I}(2,n)$ in the coordinate ring $\mathbb{K}[\plucker]$. By construction, our Pl\"ucker vector resides in the common zero locus of these relations - namely, in the variety defined by the Pl\"ucker ideal.

Sampling from $\Gr(2,n)$ is incredibly simple computationally, and is shown in Algorithm \ref{algo: sample Gr(2,n)}. Algorithm \ref{algo: sample Gr(2,n)} performs a QR decomposition on an $n \times 2$ matrix with i.i.d.\ standard Gaussian entries, 
and then takes the column space of the orthonormalised matrix, yielding a random $2$-plane in $\R^n$ 
that is distributed uniformly with respect to the Haar measure on $\Gr(2,n)$.
\begin{algorithm}
\caption{Sampling a random element of $\Gr(2,n)$}
\begin{algorithmic}[1]
\Require Integer $n \geq 2$
\Ensure An orthonormal $n \times 2$ matrix $Q$ representing a point on $\Gr(2,n)$
\State Sample $A \in \R^{n \times 2}$ with entries i.i.d.\ $\mathcal{N}(0,1)$
\State Compute the reduced QR decomposition $V = QR$
\State \Return $Q$
\end{algorithmic}
\label{algo: sample Gr(2,n)}
\end{algorithm}

\subsection*{The Tropical Grassmannian}

With a basic understanding of the Grassmannian $\Gr(2,n)$, we now seek to transform it. In the seminal work of \cite{Speyer2004-bg}, the authors introduced a tropicalized version of the Grassmannian and prove this bijects to the space of phylogenetic trees and is isomorphic to the Billera-Holmes-Vogtmann (BHV) space \cite{Billera2001-il} (which we introduced earlier, which also captures the space of phylogenetic trees). To the unfamiliar reader, the tropical semiring (max-plus algebra), denoted as $\Tmax :=(\R \cup \{-\infty\}, \oplus=\max, \otimes=+)$, has elements consisting of the real numbers and one additional value, $-\infty$, to represent ``zero.'' In the tropical semiring the usual addition and multiplication operations are replaced by:
\begin{equation}\label{eq:MaxPlus_Algebra}
    a \oplus b := \max(a, b), \quad a \otimes b := a + b,
\end{equation}
which explains why we refer to $-\infty$ as the zero-element ($-\infty$ is absorbing in max-plus multiplication).
The \emph{tropical} Grassmannian, denoted $\tropGr(2,n)$, is the \emph{tropicalization} of $\Gr(2,n)$. Since $\Gr(2,n)$ is the variety defined by the quadratic Pl\"ucker relations, this amounts to tropicalizing the quadratic polynomials in \ref{eq:plucker_relations} and intersecting the resulting tropical hypersurfaces. This process requires that we operate over a field $\mathbb{K}$ with a \emph{non-Archimedean valuation}. A valuation, loosely speaking, is a way to measure the “size” of elements in a field that respects multiplication and addition. The choice of valuation and its computational realization is more important than it may initially appear, and it presents interesting opportunities for future work. We provide additional details below, though readers may skip this section if desired. Since our (tropical) multiplication and addition differ from the standard conventions, we expect that this valuation to induce an absolute value $\NAAV{\cdot}$ also different from the usual one. This is indeed the case. 

\begin{definition}[A Field with Valuation] Consider a field $\mathbb{K}$, along with a map $\val:\K\rightarrow \R\cup\{-\infty\}$ such that for any $a,b\in\mathbb{K}$ the following axioms hold:
\begin{equation}
\begin{aligned}
\text{1. } & \val(a) = -\infty \iff a=0,\\
\text{2. } & \val(ab) = \val(a) + \val(b), \\
\text{3. } & \val(a + b) \le \max \big(\val(a), \val(b)\big).
\end{aligned}
\label{eq:valuation}
\end{equation}
Equipped with such a map, we say that $\mathbb{K}$ is a field with valuation, or a \emph{valuated field}.
\end{definition}
At first glance, $\val$ appears to operate just like a logarithm ($\log$), for which the above axioms hold over $a,b\in\mathbb{R}_{\geq0}$. Unfortunately, the restrictive domain of $\log$ disqualifies it from being a valuation in its own right, and an attempt at composition via $\log(|a|)$ (where $\lvert \cdot \rvert$ is the usual absolute value) fails the third axiom. Nonetheless, the intuition gained by thinking of valuations as a form of logarithmic mapping is the right one. Indeed, since in $\Tmax$ every element is greater than or equal to zero (i.e., $-\infty$), the underlying set of our semiring is non-negative in the true mathematical sense. Additionally, a direct application of the axioms shows that $\val(1)=\val(1^2)=2\val(1)$, which implies $\val(1)=0$. Then $\val((-1)^2)=\val(1)$ means $\val(-1)=0$ also. Thus, $\val(-a)=\val(-1)+\val(a)=\val(a)$ for all $a\in \K$.

As it turns out, our initial guess using $\log(|a|)$ was not far off. In fact, $\val(a):=\log(\NAAV{a})$ provides an equivalent definition for valuation given a properly defined \emph{non-Archimedean} absolute value $\NAAV{\cdot}$.

\begin{definition}[Non-Archimedean Absolute Value]
For a field $\mathbb{K}$, we say that $\NAAV{\cdot}$ defines a \emph{non-Archimedean absolute value} on $\mathbb{K}$ if the following hold over any $a,b\in\mathbb{K}$:
\begin{equation}
\begin{aligned}
\text{1. }\ & \NAAV{a} = 0 \ \Longleftrightarrow\ a=0,\\
\text{2. }\ & \NAAV{ab} = \NAAV{a}\NAAV{b},\\
\text{3. }\ & \NAAV{a+b} \le \max\big(\NAAV{a},\NAAV{b}\big).
\end{aligned}
\label{eq:inducedNorm}
\end{equation}
\end{definition}

A few remarks are in order. There is a clear relation between the valuation axioms and the conditions on $\NAAV{\cdot}$ just listed. In fact, any valuation \emph{induces} such a set of conditions. We see that
\[
\NAAV{a} \ :=\ \exp\big(\val(a)\big), 
\]
follows directly from our earlier characterization of $\val(a):=\log(\NAAV{a})$. The third condition is a strengthening of the triangle inequality and characterizes norms that are also referred to as \emph{non-Archimedean}. More importantly, this norm means our field $\mathbb{K}$ is a metric space with the distance between any two elements defined by $\dist_\nu(a,b):=\NAAV{a-b}$. Such metric spaces are called \emph{ultrametric} spaces due to the ultrametric inequality $\dist_\nu(i,k)\leq \max \bigl( \dist_\nu(i,j),\dist_\nu(j,k) \bigr)$. Thus, valued fields are ultrametric spaces. It is easy to verify that for any three elements $i,j,k$, the ultrametric inequalities imply $\max \bigl( \dist_\nu(i,j),\dist_\nu(i,k),\dist_\nu(j,k) \bigr)$ is attained at least twice. Thus, every ultrametric also satisfies the 4-point condition and is a tree metric. Trees satisfying the ultrametric inequality are called \emph{equidistant trees} because they possess a unique point $\treeroot$ (called the root) located on some internal branch such that it has the same distance to every leaf, $\dist_\nu(\treeroot, i)=\dist_\nu(\treeroot, j)$ for all $i,j\in[n]$. We denote by $\Uspace_\numleaves \subset \Treespace_\numleaves$ the space of all equidistant trees on $\numleaves$ leaves.

\textbf{Max-Plus $p$-Adic Valuation}

Given everything we've just covered, the reader might (justifiably) want a concrete example of a mapping that fits all such requirements. What about $\val(a)=0$ for all $a\neq-\infty$? All axioms are clearly met, though perhaps somewhat trivially. Indeed, this particular map can be applied to any field, and is (predictably) called the \emph{trivial valuation}. Our purposes, however, will require \emph{non-trivial} valuations, for which we provide the following example. 

\begin{definition}[Max-Plus $p$-Adic Valuation]
Let $\K := \Q_p$ be the field of $p$-adic numbers\footnote{The $p$-adic numbers are the completion of $\Q$ with respect to the $p$-adic (non-Archimedean) absolute value.} for some prime $p$. Then for $x = a/b \in \Q_p$ with $a,b \in \Z$ and $b \neq 0$, the discrete\footnote{The fact that the map $\val_p(x)$ is discrete turns out not to be a serious restriction, as any map given by $(\lambda\cdot\val):\K\rightarrow \R\,\cup\,\{-\infty\}$ continues to satisfy the valuation axioms for any $\lambda>0$} max-plus $p$-adic valuation $\val_p$ is
\begin{equation}\label{eq:2_adic_val}
\val_p(x) := -\big(\val_p(a) - \val_p(b)\big), \qquad
\val_p(a) := \max\{k \in \Z_{\ge 0} \mid p^k \text{ divides } a\}, \quad \val_p(0) := -\infty,
\end{equation}
with corresponding non-Archimedean absolute value on $\Q_p$ given by
\[
\PAAV{x} := p^{val_p(x)}, \qquad \PAAV{0} := 0.
\]
\end{definition}
Readers familiar with the $p$-adics may wonder about the negative sign in the definition of $\val_p(x)$, which distinguishes it from the usual definition of $p$-adic valuation. This negation is why we call this the ``Max-Plus'' $p$-adic valuation, and it is a consequence of our decision to operate using max-plus algebra (as opposed to the min-plus algebra where the $p$-adics typically operate). Loosely speaking, this exchange of semirings can be thought of exchanging the roles of the numerator and denominator for $x$. Switching from min-plus to max-plus (or vice versa) is a straightforward isomorphism given by $\min(a)=-\max(-a)$ where we remind the reader $\val(-a)=\val(a)$ applies to any valuation. 

\begin{example}
For $p=3$, we have the following three-adic (max-plus) valuations and norms:
\begin{alignat*}{2}
    \val_p(9) &= \val_p(9/1) = -\big(\val_p(3^2)-\val_p(1)\big) = -(2-0) = -2,
    &\qquad \PAAV{9} &= p^{\val_p(9)} = 3^{-2} = 1/9, \\[4pt]
    \val_p\!\left(\tfrac{2}{3}\right) &= -\big(\val_p(2)-\val_p(3)\big) = -(0-1) = 1,
    &\qquad \PAAV{\tfrac{2}{3}} &= p^{\val_p(2/3)} = 3^1 = 3.
\end{alignat*}
\end{example}

Finally, we note that valuations can be applied in a vector space $\K^n$ in the natural way by applying $\val$ component-wise every element of a vector, that is $\val(\x)=\left(\val(x_1),\ldots,\val(x_n)\right)$.

In summary, a point in $\Gr(2,n)$ can be represented by a $n \times 2$ matrix whose Pl\"ucker coordinates measure linear dependence via the determinants. When we apply a $2$-adic valuation we no longer measure the magnitude of these determinants in real geometry, but rather their \emph{order of vanishing}, that is, how many powers of $2$ divide them. A large valuation means that a point is highly divisible by $2$, indicating that columns are almost linearly dependent in a $2$-adic sense, and therefore ``close'' in a hierarchical manner. Interpreting valuations as a distance turns this hierarchical closeness into a type of branching depth or cluster, with pairs with large valuations merging deep in the tree, while pairs with small valuations separate earlier in the tree. A phylogenetic tree is, therefore, the combinatorial shadow that emerges when the geometry of a point in the Grassmannian is collapsed via a valuation.

\textbf{Tropicalization.} We now arrive at the crux of how we ``tropicalize'' $\Gr(2,n)$. For $X\subseteq \Gr(2,n)$, the Fundamental Theorem of Tropical Algebraic Geometry tells us that the \emph{tropicalization of $X$}, $\operatorname{trop}(X)$, is equal to the closure of the coordinate-wise valuation of all points in $X$, i.e.,
\begin{equation}
    \operatorname{trop}(X)=\val(X):=\{\val(\x) \mid \x \in X\}.
\end{equation}
We refer to \cite{maclagan2015introduction} for details on the fundamental theorem. 
Recalling that the Pl\"ucker coordinates of $\Gr(2,n)$ are defined by the variety of the Pl\"ucker ideal $\mathcal{I}(2,n)$, we see that the tropical Grassmannian is the \emph{tropical variety} $\operatorname{trop}(X)$, where $X=V\bigl(\mathcal{I}(2,n)\bigr)$ is the variety of the Pl\"ucker ideal. Thus, by starting from any $n\times2$ matrix $V$ of rank $2$ (whose column vectors represent the span of a two-dimensional subspace $X\subset \R^n$), we can get a point in $\tropGr(2,n)$ by taking the valuation of each coordinate of the associated Pl\"ucker vector $\plucker$. Putting this all together yields the following (max-plus) tropicalization map
\begin{equation}\label{eq:trop_plucker_map}
\tau:\ \Gr(2,n)(\K)\ \longrightarrow\ \tropGr(2,n),
\qquad
V \mapsto \tau(V) := \val(\plucker),
\end{equation}
where $\plucker$ is the associated Pl\"ucker vector of $V$ from \cref{eq:plucker_computation}. Applying the valuation axioms to the image of this map (the coordinates of the 
\emph{tropical Pl\"ucker vector}) reveals a tropical version of the Pl\"ucker relations 
that we call the \emph{tropical Pl\"ucker relations}, which can be stated as follows.
For any distinct indices $i,j,k,l \in [n]$, the classical Pl\"ucker relation
\begin{align}
\plucker_{ij}\plucker_{kl} - \plucker_{ik}\plucker_{jl} + \plucker_{il}\plucker_{jk} &= 0,
\end{align}
tropicalizes to
\begin{align}
&\val(\plucker_{ij})\otimes\val(\plucker_{kl})\ \oplus\ \val(\plucker_{ik})\otimes\val(\plucker_{jl})\ \oplus\ \val(\plucker_{il})\otimes\val(\plucker_{jk}) \\
&=\max\big( \val(\plucker_{ij}) + \val(\plucker_{kl}),\;
            \val(\plucker_{ik}) + \val(\plucker_{jl}),\;
            \val(\plucker_{il}) + \val(\plucker_{jk}) \big),
\end{align}
is attained at least twice. It should be clear that these tropical Pl\"ucker relations are equivalent to those given by the four-point condition in \cref{thm:four_point_condition}, which means that the image of our tropicalization map $\tau$ is contained in tree space $\Treespace$. Stated simply, starting from a point in $\Gr(2,n)$, taking its Pl\"ucker relations and sending it to tropical space results in a tree metric, or a tree.

Note that the multiplicative homogeneity of the Pl\"ucker relations tropicalizes to additive homogeneity of the tropical Pl\"ucker relations. Concretely, if $V\in \K^{n\times 2}$ has columns $v_1,v_2$ and we rescale them by $\lambda_1,\lambda_2\in \K^\ast$, i.e.
$V'=(\lambda_1 v_1\,\,\,\lambda_2 v_2)=V\mathrm{diag}(\lambda_1,\lambda_2)$, then every Pl\"ucker coordinate scales by the same factor:
$\plucker_{ij}'=\lambda_1\lambda_2\,\plucker_{ij}$. Hence, for all $i<j$,
\[
\val(\plucker_{ij}')=\val(\plucker_{ij})+\val(\lambda_1)+\val(\lambda_2),
\]
or equivalently,
\[
\val(\plucker')=\val(\plucker)+\shift\,\mathbf{1},
\qquad \shift:=\val(\lambda_1)+\val(\lambda_2),
\]
where $\mathbf{1}$ is the all-ones vector in $\R^{\binom{n}{2}}$. In particular, $\R\mathbf{1}\subseteq \lin(\tropGr(2,n))$, the lineality space\footnote{Another way to arrive at this conclusion is to consider the special case of ultrametrics. For any ultrametric vector $u$ it is easy to see that $u+\shift\ones$ for any $\shift\in \R$ is also an ultrametric and that this does not hold for any other vector except $\ones$. Thus $\ones$ \emph{defines} the lineality space for the space of ultrametric trees. Since ultrametric tree space is contained in tree space, $\ones$ must be an element of $\lin(\Uspace)$.} of $\tropGr(2,n)$.
Working in the tropical projective torus, $\TPS{m-1} \coloneqq \mathbb{R}^m / \mathbb{R}\mathbf{1},$ 
where $\mathbf{1} = (1,\dots,1) \in \mathbb{R}^m$ and $m = \binom{n}{2}$, we identify vectors that differ by $\shift\,\mathbf{1}$; thus we may choose $\shift$ so that $d:=\val(\plucker)+\shift\mathbf{1}$ has  positive coordinates. Since adding $\shift\mathbf{1}$ preserves the tree-metric relations, $d$ is again a tree metric for any $\shift$.

Since any $\shift$ used in $\dvec=\val(\plucker)+\shift\ones\in\TPT{\numpairs}$ still represents the same point in $\tropGr(2,n)$ it is common practice to ``normalize'' such vectors in order to establish unique representatives. The \emph{canonical coordinates}, for example, select the largest component-wise vector for $\dvec$ such that $\dvec_{ij}=0$ in at least one coordinate. Other common normalization schemes include enforcing $\dvec_{ij}=0$ for a \textit{particular} $(i,j)$ index of $\edgeset$, or enforcing a mean-zero gauge such as $\sum_{(i,j)\in\edgeset} \dvec_{ij}=0$.

We are now ready to present our first algorithm for sampling from the space of phylogenetic trees. The idea is simple: take any rational $n\times 2$ matrix $V$ of rank 2, compute its Pl\"ucker embedding $\plucker$, apply a valuation component-wise to $\plucker$, and translate the resulting vector into the positive orthant as required. Here we utilize a $2$-adic valuation. The procedure is provided in \Cref{alg:pAdic_valuation}.

\begin{algorithm}[h!]
\caption{Sampling phylogenetic trees on $\numleaves$ leaves via the $2$-adic (non-Archimedean) valuation.}
\label{alg:pAdic_valuation}
\begin{algorithmic}[1]
\Require An $\numleaves \times 2$ matrix $V=(\vecOne\,\,\,\vecTwo)$ with $\vecOne,\vecTwo\in\mathbb{Q}^\numleaves$, and a scalar $\shift\in\mathbb{R}_{>0}$.
\Ensure A tree with non-negative edge lengths representing a point in $\tropGr(2,\numleaves)$.
\State Initialize $n\times n$ distance matrix $\Dmat$: $\Dmat_{ij} \gets 0 \,\,\,\forall (i,j\in[\numleaves])$.
    \State \textbf{(Compute Pl\"ucker coordinate)} $\plucker_{ij} \gets \vecOne_i \vecTwo_j - \vecOne_j \vecTwo_i$ for $1\le i<j\le \numleaves$. \RightComment{(\Cref{eq:plucker_computation})}
    \State \textbf{(Tropicalize $\plucker_{ij}$ via valuation)} $\Dmat_{ij} = \Dmat_{ji} \gets \val_2(\plucker_{ij})$ for $1\le i<j\le \numleaves$.\RightComment{(\Cref{eq:2_adic_val})}
\State \textbf{(Construct Tree from $\Dmat$)}$~ T \gets \operatorname{NJ}(\Dmat)$.\RightComment{(Neighbour Joining \cite{Gascuel2006-sr})}\label{Step:Tree_from_D}
\State \textbf{(Check Min Branch Length)} $\eta \gets \min_\mathcal{E}w_\mathcal{E}$.
\If{$\eta \leq 0$} 
    \State $\Dmat_{ij} \gets \Dmat_{ij}+\epsilon \quad \forall(i\neq j)$.
    \State Go to \Cref{Step:Tree_from_D}.
\EndIf
\State \Return $T$
\end{algorithmic}
\end{algorithm}

\begin{remark}
Tree metrics correspond to finite distances. If any Pl\"ucker minors vanish (i.e. $\plucker_{ij}=0$), the tropicalization map $\val(\plucker)$ will yield infinite valuations on these indices. Depending on how $V$ is sampled, this may not be a practical issue. Regardless, to avoid such issues one can (i) design inputs to avoid total cancellations, or (ii) “jitter” coefficients to break ties.
\end{remark}

P-adic valuations, while conceptually simple, present non-trivial computational challenges, and working over such fields can be intractable for many real-world applications. At this point, one might wonder about the Pl\"ucker embedding itself and whether, by computing determinants tropically, one can attain a tropical Pl\"ucker vector directly from a $n\times 2$ matrix $V$ rather than tropicalizing (by applying coordinatewise valuations) the Pl\"ucker vector attained from the usual embedding. If we define $\tilde{\plucker}_{ij}:=\vecOne_i\otimes \vecTwo_j \oplus \vecOne_j\otimes \vecTwo_i$ we see that the resulting vector $\tilde{\plucker}$ satisfies the tropical Pl\"ucker relations by construction, and thus also satisfies the 4-point condition. Such a construction is called the \emph{tropical Stiefel map} \cite{Fink_2015_Steifel}. Then by appropriate translation into the positive orthant as before we arrive at a tree metric! Unfortunately, the authors of \cite{Fink_2015_Steifel} show that this map is only capable of producing caterpillar trees and as such, provides little utility for our purposes here. Figure \ref{fig:adic_vs_stiefel} provides a comparison of the trees produced using the method of \cref{alg:pAdic_valuation} (left) and the tropical Stiefel map (right). Note the use of neighbour joining in \cref{alg:pAdic_valuation} is guaranteed to return a unique tree as the distance matrix is tree metric\cite{Gascuel2006-sr}.

\begin{figure}[htb]
  \centering
  \begin{minipage}{0.48\textwidth}
    \centering
    \includegraphics[width=\linewidth]{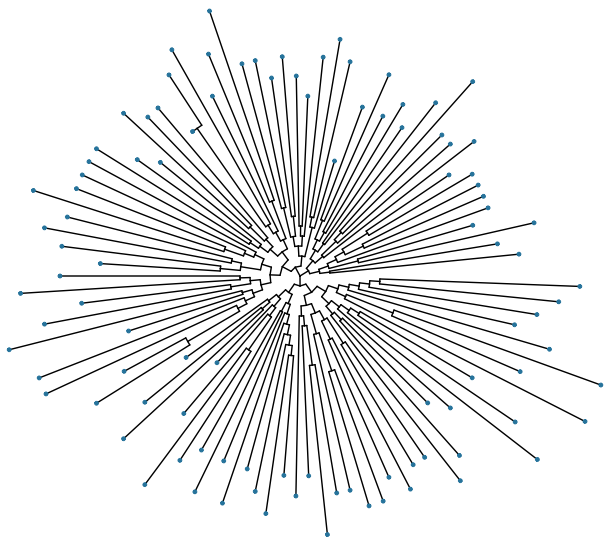}
  \end{minipage}\hfill
  \begin{minipage}{0.48\textwidth}
    \centering
    \includegraphics[width=\linewidth]{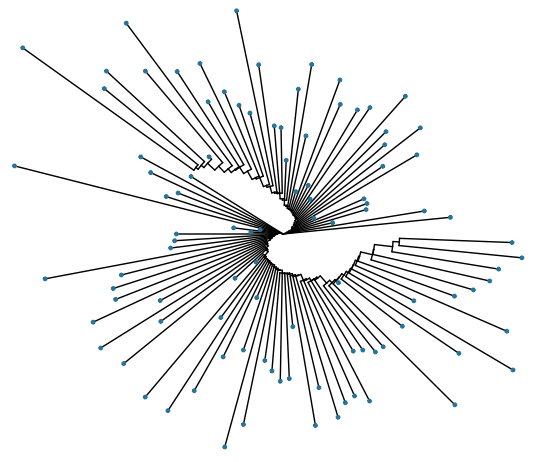}
  \end{minipage}
  \caption{Two phylogenetic trees with $\numleaves=100$ leaves obtained using the same $\numleaves \times 2$ matrix $V$. The tree on the left is obtained by tropicalizing the Pl\"ucker vector using a 2-adic valuation in the manner of Algorithm \ref{alg:pAdic_valuation}. The tree on the right is obtained using a tropical Stiefel map.}
  \label{fig:adic_vs_stiefel}
\end{figure}

\subsection*{The Structure of Tree Space}

Take the $\numleaves \times \numleaves$ matrix $D$ that represents some phylogenetic tree $\mathcal{T}$ on $\numleaves$ leaves, and consider $D$ as a weighted adjacency matrix for a graph $G$. Since every pair of leaves in $\Treespace_\numleaves$ is separated by some finite distance, we can assume that $D$ contains all finite entries, and so $G$ is the complete graph $K_\numleaves$ on $\numleaves$ nodes. Any (non-repeating) subset of at least three nodes in $G=K_\numleaves$ ($G$ is an undirected graph) forms a simple cycle $C$ by returning to the starting node, and the smallest such cycles are triangles of the form $C=\{i,j,k\}$ with associated edges weights $\{d_{ij},\,d_{ik},\,d_{jk}\}$. By \cref{eq:inducedNorm}, such triangles must have edge weights that achieve the maximum at least twice, and it is easy to show that this extends to any cycle of $G$. 

The condition imposed by the ultrametric inequality on the circuits of $G$ is an analogous way of defining that matrices $D$ that are contained in tree space. It follows that these ``ultrametric circuits'' form an equivalent definition of tree space. This is precisely the observation made by Ardila in \cite{Ardilasubdominant}, where the correspondence between $K_\numleaves$ and equidistant trees was made explicit. These ``dependencies'' amongst the edge weights of the circuits of $G$ define a matroid\footnote{For the reader who is curious about matroids the authors recommend \cite{gordon2012matroids}, which provides a delightful introduction to the topic. See also \cite{oxley2006matroid} for more rigorous treatment.} $M:=M(\edgeset,C)$ where $\edgeset$ is called the \emph{ground set} and $C=\{C_1,\ldots,C_N \mid C_k\subseteq \edgeset\}$ is the set of circuits. In particular, matroids that are encoded as graphs (where edges of $G$ define the ground set) are called \emph{graphical matroids} $M(G)$. Thus, the matroid we are interested in is $M(K_\numleaves)$ and explains our use of $\edgeset$ to denote the ground set. This shows precisely why Dress used the term \emph{valuated matroids} in describing these spaces - the tropicalization of $M(K_\numleaves)$, $\operatorname{trop}(M(K_\numleaves))$ defines the space of ultrametric trees.

At this point the reader might wonder, ``If valued fields are ultrametric spaces, why doesn't a tropicalization map (such as in \cref{alg:pAdic_valuation}) result in an ultrametric (i.e. an equidistant tree)?'' In fact, such maps can be specialized to yield ultrametrics, but this need not be the case in general, as the following example demonstrates.
\begin{example}\label{ex:not_ultrametric}
    Consider $V=(\vecOne\,\,\,\vecTwo)$ given by $\vecOne=(16,8,4,2)^\top$ and $\vecTwo=(0,1,1,2)^{\top}$ and the $2$-adic max-plus valuation. If we define $\Delta \vecOne:=(d_\nu(\vecOne_i,\vecOne_j)\mid1\leq i<j\leq4)$ it is easy to see that $\Delta \vecOne=2^y$ with $y=-(3,2,1,2,1,1)$ is an ultrametric. At the same time, $\plucker=(16,16,32,4,14,6)^\top$ yields $\val(\plucker)=-(4,4,5,2,1,1)^\top$, which isn't an ultrametric. If we change $\vecTwo$ so that  $\vecTwo=\ones$, then $\plucker=(8,12,14,4,6,2)^\top$, and we see that now $\val(\plucker)=y$ is now an ultrametric.
\end{example}
This example hints at how we might better understand the lineality space of tree space $L:=\lin(\Treespace)$ and how it differs from the (simpler) lineality space given by equidistant tree space (which is just $\lin(\Uspace)=\ones$). In fact, this distinction is precisely what distinguishes one space from the other. 

Recall our parameterization of $V=(\lambda_1 \vecOne\,\,\,\lambda_2 \vecTwo)$ given by scaling the columns of $V$ by $\lambda_1,\lambda_2\neq0$. Now, suppose that instead of performing column scaling, we instead perform row scaling according
\begin{equation*}
    V'=\begin{pmatrix}
        \mu_1\vecOne_1 & \mu_1\vecTwo_1\\
        \vdots & \vdots \\
        \mu_n\vecOne_n & \mu_n\vecTwo_n
    \end{pmatrix} .
\end{equation*}
Following our tropicalization map as before, we have
$\plucker'_{ij}=\mu_i\mu_j\plucker_{ij}$ so that $\val(\plucker')_{ij}=\val(\plucker_{ij})+\val(\mu_j)+\val(\mu_i)$. Unlike when we scaled column vectors, here the additive scaling is not uniform across the coordinates of $\val(\plucker')$. That is, we have $\val(\plucker')=\val(\plucker)+\val(\mu)$ where $\mu:=(\mu_i\mu_j \mid 1\leq i<j \leq n)$. In this way, we can think of the $L_n=\lin(\tropGr(2,n))$ as the image of our tropicalization map when column \emph{and} row scaling is performed on the input matrix $V$.

Our explanation of $L$ thus far was designed to gain intuition, and as such, was rather informal. We now proceed a bit more formally, following \cite{maclagan2015introduction}, in order to properly define it. 

\textbf{Lineality space.} The lineality space $L_n$ results from the homogeneity of the Pl\"ucker ideal $\mathcal{I}(2,n)$ with respect to the $\mathbb{Z}^m$-grading given by $\deg(\plucker_{ij})=e_i+e_j \in \mathbb{Z}^m$. To see this, consider an arbitrary Pl\"ucker relation indexed by $\{i,j,k,l \}$. Then $\deg(\plucker_{ij}\plucker_{kl})=(e_i+e_j)+(e_k+e_l)=e_i+e_j+e_k+e_l$ is the same for each term in the relation. Thus, for the linear map 
\[
\phi:\mathbb{R}^n \rightarrow \mathbb{R}^{\binom{n}{2}}, \quad (\vecOne_1,\ldots,\vecOne_n) \mapsto (\vecOne_i+\vecOne_j)_{1\leq i<j\leq n},
\]
we get that for any vector $\tilde{\plucker}\in \tropGr(2,n)$, if we set $\tilde{\plucker}':=\tilde{\plucker}+\phi(\vecOne)$, then $\plucker'_{ij}=\plucker_{ij}+(\vecOne_i+\vecOne_j)$ in each coordinate entry. The corresponding \emph{tropical} Pl\"ucker relations for $\tilde{\plucker}'$ have terms that look like $\plucker'_{ij}+\plucker'_{kl}=(\plucker_{ij}+\vecOne_i+\vecOne_j)+(\plucker_{kl}+\vecOne_k+\vecOne_l)=(\plucker_{ij}+\plucker_{kl})+(\vecOne_i+\vecOne_j+\vecOne_k+\vecOne_l)$. Clearly, all terms in the relation are shifted by the same additive constant so that the maxima remain unchanged. Since this is true of every relation, $\operatorname{Trop(X)}$ (the tropical variety) is invariant under translation by $\phi(\vecOne)$. In other words, $L$ is the image of the linear map $\phi$, i.e., 
\begin{equation} \label{eq:lineality_space}
L_n=\lin(\tropGr(2,n))=\mathrm{im}(\phi)
=\textup{span}\Bigl\{\sum_{j\neq i} e_{ij}\ :\ 1\le i\le n\Bigr\}
\subseteq \R^{\binom n2}. 
\end{equation} 
Moreover, under the standard identification of $\tropGr(2,n)$ with the space of tree metrics on $[n]$,
the lineality spaces agree:
\[
\lin(\Treespace_n)=\lin(\tropGr(2,n)) 
\]
Here $\Treespace_n\subseteq \R^{\binom{n}{2}}$ denotes the space of (phylogenetic) tree metrics on the leaf
set $[n]=\{1,\dots,n\}$, with coordinates indexed by pairs $(i,j)$, $1\le i<j\le n$.
Examining \cref{eq:lineality_space}, the vector $\sum_{j\neq i}e_{ij}$ is the characteristic vector of the
set of edges incident to node $i$ in the complete graph $K_n$.  Equivalently, $L_n$ is the row space of the
node-edge incidence matrix of $K_n$.  Under the tropicalization map, this is precisely the pattern of
coordinates modified by rescaling the $i$th row of $V$.  Intuitively, these perturbations correspond to
changing the pendant length at leaf $i$ while leaving the underlying tree topology unchanged.

Note that each of these characteristic vectors also corresponds to the edges of a cut-sets given by the partition $\{i\}\uplus \{[\numleaves]\setminus i\}$. Returning to our matroidal perspective, these sets of edges form what are called \emph{cocircuits} of the matroid $M(K_\numleaves)$ and (by definition) are circuits of the dual matroid $M(K_\numleaves)^*$. This highlights a duality relation for graphical matroids that is well-known in network optimization.

We are now ready to state an important decomposition of tree space, which allows us to consider any tree as the sum of an ultrametric and an element from the lineality space.
\begin{lemma}
[{\cite[Lemma 4.3.9]{maclagan2015introduction}}]
\label{lem:decomposition}
    Every tree metric $\dvec \in \TPT{\binom{\numleaves}{2}}$ is an ultrametric $u\in \Uspace_\numleaves$ plus a vector in the lineality space $L_\numleaves$. Thus, the space of phylogenetic trees has the decomposition
    \begin{equation}
        \Treespace_\numleaves=\operatorname{trop}(M(K_\numleaves)) + L_\numleaves.
    \end{equation}
\end{lemma}

\begin{remark}
The decomposition given in \Cref{lem:decomposition} suggests that sampling from $\Treespace$ can be accomplished by sampling separately from both components of the decomposition and then adding them in a manner that returns a vector in the positive orthant. Since the ultrametric portion determines the clade partitions, one can think of this component as fixing the tree's topology while the lineality space adjusts particular branch lengths without modifying clade groupings. In particular, we can modify \Cref{alg:pAdic_valuation} by requiring that $\vecTwo=\ones$. This results in \Cref{alg:pAdic_valuation} sampling ultrametrics, and can provide some computational advantages.
\end{remark}
\vspace{4pt}
\begin{remark}[Distribution of $2$-adic valuations of random differences and its effect on tree shape]
Let $K\in\mathbb{N}$ and let $\vecOne,\vecTwo$ be independent and uniform on $\{0,1,\dots,2^K\!-\!1\}$. Set $D\equiv \vecOne-\vecTwo\pmod{2^K}$. Since subtraction is a bijection of the finite abelian group $\mathbb{Z}/2^K\mathbb{Z}$, the random variable $D$ is uniform on $\mathbb{Z}/2^K\mathbb{Z}$. For $0\le k<K$, the event $\val_2(D)=k$ means that $D$ is divisible by $2^k$ but not by $2^{k+1}$; these are exactly the residues $D=2^k(2m+1)$ with $0\le m<2^{K-k-1}$. Hence
\[
Pr\bigl(\val_2(D)=k\bigr)\;=\;\frac{2^{K-k-1}}{2^K}\;=\;2^{-(k+1)},
\qquad k=0,1,\dots,K-1,
\]
and $\mathbb{P}\bigl(\val_2(D)\ge K\bigr)=2^{-K}$. Thus $\val_2(D)$ is a geometric distribution truncated at $K$. 

If instead $u,v$ are sampled independently and uniformly from $\{1,\dots,M\}$ with $M$ large, then residue classes modulo $2^k$ are approximately equidistributed; in particular, for fixed $k$ with $2^k\ll M$,
\[
\mathbb{P}\bigl(\val_2(u-v)=k\bigr)\;\approx\;2^{-(k+1)}.
\]
Consequently, for a {\em single} pair $(i,j)$ of leaves, the marginal distribution of the $2$-adic dissimilarity $\dvec_{ij}:=\val_2(y_i-y_j)$ is (approximately) geometric. Across many pairs these values are \emph{not} independent, but there is still a strong tendency for $d_{ij}$ to take only a few small integer values. This creates large blocks of exact ties among quartet sums in the four-point test, which \emph{on average} promotes more balanced topologies (many clades coalescing at the same level). 

It is important to note that when we say that pairwise distances are geometric, this refers only to the marginal distribution of $\val_2(y_i-y_j)$ for a fixed pair. These values are not independent, and must satisfy the constraints of a tree metric.
\end{remark}
\vspace{4pt}
\begin{example}[$2$-adic ladder (caterpillar) from powers of $2$]
Let $y_i=2^{\,i-1}$ for $i=1,\dots,n$ and define $\dvec_{ij}:=\val_2(y_j-y_i)$ for $i<j$. Then
\[
y_j-y_i \;=\; 2^{i-1}\bigl(2^{\,j-i}-1\bigr),\qquad\text{with }2^{\,j-i}-1\text{ odd},
\]
so $\val_2(y_j-y_i)=i-1$. In particular,
\[
\dvec_{12}=0,\quad \dvec_{23}=1,\quad \dvec_{34}=2,\;\dots,\; \dvec_{(n-1)n}=n-2,
\]
and more generally $\dvec_{ij}=i-1$ for every $i<j$. The $2$-adic ultrametric inequality
\(
\val_2(x-z)\ge \min\{\val_2(x-y),\val_2(y-z)\}
\)
is satisfied, and the hierarchical pattern is strictly nested: the pair $\{n-1,n\}$ coalesces deepest, then $\{n-2,\{n-1,n\}\}$, and so on. The resulting rooted tree is a maximally imbalanced \emph{caterpillar} (a single backbone with leaves attaching one by one at increasing depth). 
\end{example}

\section*{Sampling over $\tropGr(2,n)$}

Numerous methods exist for sampling from tree space, yet vanishingly few can generate samples that are uniformly distributed across the entire space. Tree space is typically characterised as ranging between two structural extremes: the star tree and the caterpillar (ladder) tree. Perhaps the simplest approach is the Proportional to Distinguishable Arrangements (PDA) model, which assumes each labelled topology is equally likely. Representing trees via a bijection to integer vectors (e.g. \cite{Penn2024-ja}), one can equivalently view PDA as uniform sampling in that vector space. Maximum-likelihood tree inference \cite{Minh2020-ey} does not explicitly define a prior on topologies; however, in a Bayesian interpretation, ML is equivalent to maximum a posteriori estimation under a uniform (PDA) prior. PDA produces trees that are highly unbalanced on average, but still under-represents extremely caterpillar-like structures.

Model-based formulations such as the Yule (pure-birth) process \cite{Yule1925-de} and the general birth–death process \cite{Kendall1948-ht} provide well-defined probability distributions over trees and are commonly used to sample random phylogenies. These models tend to generate more balanced tree shapes than PDA, but newer age-dependent generalisations \cite{Andersen2025-rb} allow interpolation across a broad spectrum of shapes, from highly unbalanced to highly balanced trees. Nonetheless, these processes do not induce uniform sampling over tree space.

Tropical tree space offers new avenues for exploring and sampling phylogenetic tree distributions. From our basic understanding of tree space itself, we want to consider collections of trees, i.e., multiple points in the tropical Grassmannian. For this, we will primarily focus on the specific case of ultrametric (or equidistant) trees, the reasons for which will soon be made clear. For now, let us simply use $\Uspace\subset\Treespace$ to denote the \emph{space of ultrametrics} or, equivalently, the \emph{space of equidistant trees}. In the notation of \Cref{lem:decomposition} we have that 
\[
\Uspace_n \;\cong\; \bergFan(M(K_n)) \;=\; \operatorname{trop}(M(K_n)),
\]
where $\bergFan$ is called the \emph{Bergman fan} of the matroid $M(K_n)$. The Bergman fan of a matroid is the tropical linear space associated with that matroid. 

A common question is how to define a measure of closeness between two trees. The Billera–Holmes–Vogtmann (BHV) distance is one of the most widely used approaches where the distance is both mathematically sound and biologically meaningful; however, its computation is costly and does not always scale well to large datasets \cite{Owen2011-sm}. The Robinson–Foulds (RF) distance \cite{Robinson1981-nd}, by contrast, is computationally cheap and widely used in practice, but it only captures topological differences (the presence or absence of splits) and ignores branch lengths, which limits its sensitivity. RF distance also saturates quickly and trees can attain a maximal normalised distance of one, despite being quite similar.  The Subtree-Prune and Regraft (SPR) distance is another alternative that reflects the minimal number of topological rearrangements needed to transform one tree into another, which ties directly to biologically meaningful evolutionary processes; however, it is NP-hard to compute exactly and often requires heuristics, reducing its practicality for very large trees. 

\subsubsection*{The Tropical Metric (Generalized Hilbert Projective Metric)}

The tropical Grassmannian induces a natural notion of distance between trees via the \emph{tropical metric}. 
Given two vectors \( \vecOne, \vecTwo \in \R^n \), the tropical distance is defined as
\begin{equation}
\dtr(\vecOne,\vecTwo) \;=\; \max_{i}\,(\vecOne_{i}-\vecTwo_{i}) \;-\; \min_{i}\,(\vecOne_{i}-\vecTwo_{i}).
    \label{eq: hilbert_metric}
\end{equation}
Intuitively, this distance measures the range of coordinate-wise differences between $\vecOne$ and $\vecTwo$. It is also easy to see that it is invariant under adding a constant vector to either vector, which we should expect since $\vecOne$ and $\vecOne+\shift\ones$ represent the same point in $\tropGr(2,n)$. Note, however, that $\dtr$ is \emph{not} invariant by adding \emph{any} vector from $L$. In other words, perturbing $\vecOne$ or $\vecTwo$ by an element of $L$ will, in general, change $\dtr(\vecOne,\vecTwo)$. 

Tropical distance also respects the combinatorial structure of tree space: small topological changes correspond to small tropical distances. From a biological perspective, the tropical distance between two ultrametric trees has a natural interpretation. 
If two trees, $T_1,\,T_2$, are represented by their cophenetic vectors, $\dvec^{(1)},\,\dvec^{(2)}$, where entry $\dvec^{(\cdot)}_{ij}$ records the divergence time of taxa $i$ and $j$, 
then the coordinate-wise differences $\dvec^{(1)}_{ij} - \dvec^{(2)}_{ij}$ measure how much earlier or later each pair of taxa coalesces in one tree relative to the other. Thus this captures the \emph{spread} of these discrepancies: it is the gap between the clade whose divergence is shifted the most earlier and the clade whose divergence is shifted the most later across the two trees. 
Because this metric is invariant under adding a constant to all entries, it ignores uniform shifts in divergence times (e.g.\ due to calibration), and instead reflects the relative rescaling of divergence events across lineages. 
In this sense, the tropical distance quantifies the worst-case disagreement in relative evolutionary timing, emphasizing how unevenly the two trees stretch or compress different parts of their histories.

\paragraph{Tropical Line Segments between Ultrametric Trees.}

Consider two ultrametrics $\vecOne,\vecTwo$ that are both finite, i.e., $\vecOne,\,\vecTwo \in \left\{ \Uspace_\numleaves \cap \TPT{\numpairs} \right\}$ (recall that $m=\binom{\numleaves}{2}$). The \emph{tropical line segment} between $\vecOne, \vecTwo \in \TPT{\numpairs}$ is defined as
\begin{equation}\label{eq:trop_line_seg}
    \Gamma_{\vecOne,\vecTwo} \;=\; \left\{ \ell \odot \vecOne \oplus \vecTwo \;\middle|\; \ell \in [\,\min(\vecTwo-\vecOne), \max(\vecTwo-\vecOne)\,] \right\}.
\end{equation}

Tropical line segments exist for any two points in $\TPT{\numpairs}$, which includes the finite elements of $\tropGr(2,n)$ (i.e., cophenetic vectors of trees with finite pendant edge lengths). However, for ultrametrics (and by extension for cophenetic vectors of equidistant trees), they acquire particularly nice properties. For starters, note that for any choice of $\ell$ the resulting point in $\TPT{\numpairs}$ - remains an ultrametric. That is $\Gamma_{\vecOne,\vecTwo}\subset \Uspace_n$. Thus, by interpolating among ultrametric trees, we are guaranteed to remain in (ultrametric) tree space. The reasons for this stems from the fact that $\Uspace_n$ is actually a \emph{tropical linear space} and is therefore \emph{tropically convex}, see \cite{speyer2008tropLinearSpaces, lin2017convexity}. Thus, our tropical line segment (anchored by points in a tropically convex set) is simply the tropical analogue of what we are used to when dealing with classical convexity.

The interval $[\min(\vecTwo-\vecOne), \, \max(\vecTwo-\vecOne)]$ is exactly the set of scalar shifts $\ell$ such that the coordinates of $\ell \odot \vecOne \oplus \vecTwo$ interpolate between $\vecOne$ and $\vecTwo$ in the tropical sense.  
Hence, the tropical line segment $\Gamma_{\vecOne,\vecTwo}$ is parameterized by $\ell$ in this interval. From this it becomes simple to sample trees between two anchor trees using \cref{eq:trop_line_seg}. We see from Figure \ref{fig:HAR sampling} that we can sample two ultrametric trees from a coalescent process and adjusting $\ell$ in Algorithm \cref{eq:trop_line_seg} allows us to smoothly interpolate between these two trees. As expected, the tropical distance is perfectly linear between these two trees across the line segment, but this distance is closely correlated with BHV, SPR and RF distances too. 

\begin{figure}[h]
    \centering
    \includegraphics[width=1\linewidth]{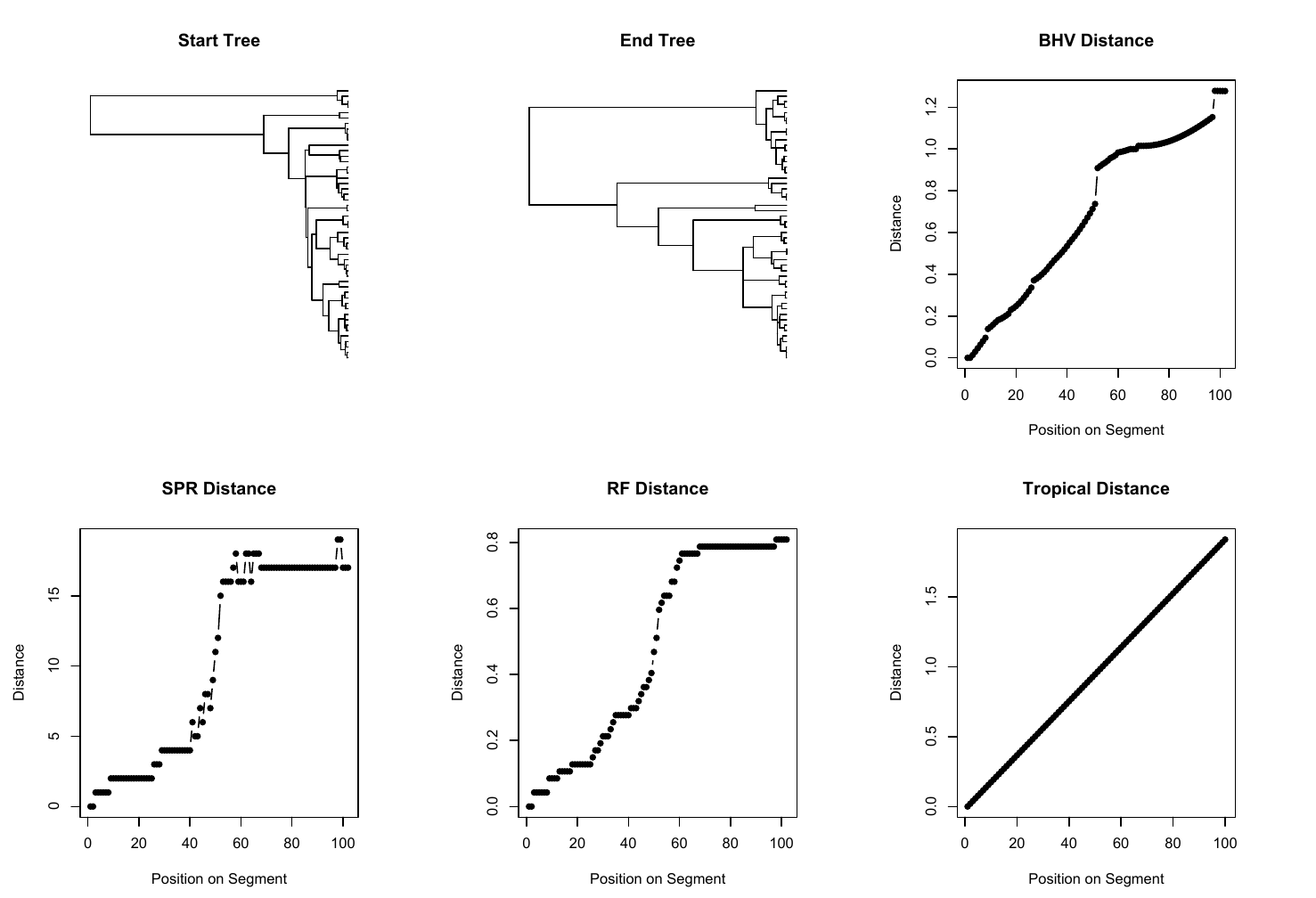}
    \caption{Comparison of start and end trees with corresponding distance trajectories.
The first panel shows the initial phylogenetic tree (left) and the final phylogenetic tree (right) inferred by UPGMA. The subsequent four panels display the evolution of tree distances across the sampled path, measured using four different metrics: Billera–Holmes–Vogtmann (BHV) distance, subtree prune-and-regraft (SPR) distance, Robinson–Foulds (RF) distance, and the tropical distance. Together, these plots illustrate both the structural differences between the start and end trees and how these distances vary along the interpolation.}
    \label{fig:HAR sampling}
\end{figure}

In \Cref{fig:HAR sampling} we demonstrate this \textit{smoothness} by constructing tropical line segments between ultrametric trees and tracking four alternative distances (BHV, SPR, RF, and tropical). We (i) verify that interpolation by $\ell\odot \vecOne \oplus \vecTwo$ remains within the ultrametric tree space, and (ii) observe that the tropical distance varies \emph{exactly linearly} along the path reflecting its tropical projective invariance. By contrast, BHV changes smoothly with occasional slope/jump transitions and SPR/RF update in discrete steps. 

\paragraph{The Tropical Convex Hull of a Set of Ultrametric Trees}

The fact that $\Uspace$ is closed when taking tropical line segments naturally extends beyond just pairs of ultrametrics. For a set of ultrametrics, we can consider their \emph{tropical convex hull}, the tropical analogue of classic convex hulls. Just like for tropical line segments, any set of ultrametric trees has its tropical convex hull also contained in $\Uspace$. 
Given a set of ultrametrics $\mathcal{S}=\{u^{(1)},\dots,u^{(k)}\} \in \TPS{m-1}$, their tropical convex hull $\operatorname{tconv}$ is defined as
\begin{equation}\label{eq:tconv}
\operatorname{tconv}\bigl(u^{(1)},\dots,u^{(k)}\bigr) 
   \;=\; \Bigl\{ \;\bigoplus_{j=1}^{k} \lambda_{j} \odot u^{(j)} 
        \;\Big|\; \lambda_{1},\dots,\lambda_{k} \in \Tmax\;\Bigr\}.
\end{equation}

Thus we can perform such ``interpolation'' on $\mathcal{S}\subset\Uspace$ without ever leaving (ultrametric) tree space. Any point $w \in \operatorname{tconv}(u^{(1)},\dots,u^{(k)})$ has coordinates
\[
w_{i} \;=\; \max_{j=1,\dots,k}\,\bigl(u^{(j)}_{i}+\lambda_{j}\bigr), 
\qquad i=1,\dots,\tbinom{n}{2}.
\]
The tropical convex hull is the smallest tropically convex set 
containing all of the input trees, and generalizes the tropical line segment 
to higher dimensions. Sampling within $\operatorname{tconv}$ thus provides
a way to interpolate among several anchor trees at once, always producing
valid ultrametric trees inside $\Uspace_n$.

Tropical convex hulls are \emph{tropical polytopes} in the tropical projective torus $\TPS{\numpairs-1} := \R^{\numpairs}/\R\mathbf{1}$. In fact, Ardila showed in \cite{Ardilasubdominant} that $\Uspace_n$ \emph{is} a tropical polytope. Note that Ardila’s result is formulated in \emph{tropical projective space} in the sense of
\[
\bigl((\R\cup\{-\infty\})^{\numpairs} \setminus \{(-\infty,\dots,-\infty)\}\bigr)
/ \R\mathbf{1},
\]
meaning, some of the extremal generators (vertices) of  $\Uspace_n$ may have
coordinates equal to $-\infty$ and therefore lie outside our torus
$\TPS{\numpairs-1}$. In this paper we adopt the convention
$\TPS{\numpairs-1} := \R^{\numpairs}/\R\mathbf{1}$ and implicitly restrict to the finite part of $\Uspace_n$, i.e.\ its intersection with the tropical projective torus; this does not affect the combinatorial structure of the tropical polytope.
These vertices of $\Uspace_n$ can be computed directly from
$\operatorname{trop}(M(K_n))$ as the image under the valuation map of the
maximal proper \emph{flats} of $M(K_n)$. By enumerating all such maximal
flats, it is (in theory\footnote{We say ``in theory'' because any algorithm for enumerating flats for $M(K_n)$ has a worst-case running time that is exponential in the input size \cite{oxley2006matroid}}, and so this quickly becomes infeasible for trees with many leaves) possible to make explicit the tropical polytope (tropical convex hull) of $\Uspace_n$.

\paragraph{Tropical Projection.} Given a set $\mathcal{S}$ of ultrametrics as before, suppose we have $w\notin\Uspace$ and want to find the closest (in the sense of the tropical metric) ultrametric to $w$ that is still contained in $\operatorname{tconv}(\mathcal{S})$. We use $\mathcal{P}$ to denote the tropical polytope, i.e. $\mathcal{P}=\operatorname{tconv}(\mathcal{S})$. The \emph{tropical projection} of $w \in \TPS{\numpairs-1}$ onto $\mathcal{P}$ is

\begin{equation}\label{eq:trop_proj}
    \pi_{\mathcal{P}}(w)
    \;=\;
    \bigoplus_{i=1}^{k} \lambda_i \odot u^{(i)},
    \qquad
    \lambda_i
    \;=\;
    \max\{\lambda\in\mathbb{R}\mid \lambda\odot u^{(i)} \le w\}
    \;=\;
    \min_{j}\bigl(w_j - u^{(i)}_j\bigr).
\end{equation}

\subsection*{Tropical Hit and Run}

Smooth interpolation along tropical line segments anchored by ultrametric trees enables the use of Markov Chain Monte Carlo techniques by sampling along tropical line segments as a subroutine. Such a scheme was proposed in \cite{Yoshida2023} under the name \emph{tropical hit and run} (HAR), where it was applied to both arbitrary tropical polytopes and the space of ultrametrics. The essence of the approach is as follows. Given a tropical polytope $\operatorname{trop}(X)$ and an initial point $x\in \operatorname{trop}(X)$, sample another point $y\in \operatorname{trop}(X)$ via some predetermined method. With $x$ and $y$, a point $z$ is sampled uniformly from the tropical line segment $\Gamma_{x,y}$ as the next proposed move. The algorithm then either accepts $z$, setting $x\leftarrow z$, or rejects it, keeping $x$ unchanged. The procedure then repeats until termination. The output is a subset of accepted moves encountered throughout the course of the algorithm.

The primary challenge in implementing tropical HAR is in the selection of $y$. For arbitrary $\operatorname{trop}(X)$ the authors of \cite{Yoshida2023} propose sampling from among the tropical vertices, or alternatively, sampling from some Euclidean space followed by tropically projecting $y$ onto $\operatorname{trop}(X)$. Both of these approaches have drawbacks, which we now discuss. The former method requires explicit representation of the vertices of $\operatorname{trop}(X)$, which can be prohibitive for trees with many leaves. In particular, it requires enumerating the cocircuits\footnote{Briefly, every matroid $M$ has a dual matroid $M^*$ over the same ground set. Circuits of $M$ correspond to \emph{cocircuits} of $M^*$ and vice versa. We use cocircuits in our earlier statement, as these are in fact the complement of (and therefore combinatorially equivalent to) the maximal flats in $M$.} of $M(K_n)$, of which there are $2^{(n-1)}-1$. Furthermore, tropical line segments with one endpoint at a vertex of $\operatorname{trop}(X)$ tends to result in a considerable portion of the samples coming from cones of the Bergman fan that are not of maximum possible dimension, i.e. ultrametrics corresponding to non-binary (unresolved) topologies.

The latter method, which performs tropical projection, also requires explicit vertex representation in the general case. Fortunately for us, specialized algorithms exist that bypass the need to explicitly compute all such vertices. In particular, Ardila showed in \cite{Ardilasubdominant} that single-linkage hierarchical clustering coincides with $\pi_{M(K_n)}$ up to translation by $\ones$, providing us a polynomial-time algorithm for projecting onto $\operatorname{trop}(M(K_n))$ without the need for explicit tropical vertices. Even so, distributions associated with sampling over a euclidean space do not translate once tropical projection is applied because such a mapping is not injective.  The overall procedure is outlined in \cref{algo:trop_HAR_ultra}, where we use $\mathcal{P}$ to denote the tropical polytope corresponding to $\Uspace_n$.

\begin{algorithm}
\caption{Tropical Hit and Run (HAR) over ultrametric tree space $\Uspace_\numleaves$ \cite{Yoshida2023}}
\label{algo:trop_HAR_ultra}
\begin{algorithmic}[1]
    \Require Desired sample size $N$ and initial point $x_0 \in \Uspace_\numleaves$.
    \Ensure $\mathcal{S}=\{u^{(1)},\ldots,u^{(N)} \} \subset \Uspace_\numleaves$.
    \State Initialize $\mathcal{S}\leftarrow \emptyset$, $x\leftarrow x_0$, and $k\gets 1$.
    \While{$k\leq N$}
        \State Sample a representative $y \in \R^{\numpairs}$ of a point in $\R^{\numpairs}/\R\mathbf{1}$,
               where $\numpairs=\binom{\numleaves}{2}$.
        \State Compute $\tilde{y}\gets \pi_\mathcal{P}(y)$ (e.g.\ via single-linkage projection onto $\Uspace_\numleaves$).
        \State Draw $\ell$ uniformly from $\bigl[\min_i(\tilde y_i-x_i),\,\max_i(\tilde y_i-x_i)\bigr]$.
        \State \textbf{Sample from $\Gamma(x,\tilde y)$:} Set $u^{(k)} \gets \ell \odot x \oplus \tilde y$
               (a point in $\R^{\numpairs}/\R\mathbf{1}$, cf. \Cref{lemma:welldefined}).
        \If{the algorithm \emph{accepts} $u^{(k)}$}  
            \State $\mathcal{S}\leftarrow \mathcal{S}\cup\{u^{(k)}\}$,\quad $x\leftarrow u^{(k)}$,\quad $k\leftarrow k+1$
        \EndIf
    \EndWhile
    \State \Return $\mathcal{S}$
\end{algorithmic}
\end{algorithm} 
\begin{lemma} \label{lemma:welldefined}[Well-definedness on the tropical projective torus]
Let $\numpairs = \binom{\numleaves}{2}$ and let
\[
\TPS{\numpairs-1} := \R^{\numpairs}/\R\ones
\]
be the tropical projective torus. In Algorithm~\ref{algo:trop_HAR_ultra}, fix
points $x, \tilde y \in \TPS{\numpairs-1}$ and choose representatives
$x,\tilde y \in \R^{\numpairs}$. For any
\[
\ell \in \bigl[\min_i(\tilde y_i-x_i),\,\max_i(\tilde y_i-x_i)\bigr],
\]
define
\[
u := \ell \odot x \oplus \tilde y \;\in\; \R^{\numpairs}.
\]
Then the class of $u$ in $\TPS{\numpairs-1}$ is independent of the choice of
representatives for $x$ and $\tilde y$. In particular, the proposal
$u^{(k)} = \ell \odot x \oplus \tilde y$ in
Algorithm~\ref{algo:trop_HAR_ultra} is well defined as a point of
$\TPS{\numpairs-1}$.
\end{lemma}

\begin{proof}
Let $x',\tilde y' \in \R^{\numpairs}$ be another choice of representatives for
the same points in $\TPS{\numpairs-1}$. Then there exists $\alpha\in\R$ such
that
\[
x' = x + \alpha\ones,
\qquad
\tilde y' = \tilde y + \alpha\ones.
\]
For each coordinate $i$ we have
\[
\tilde y'_i - x'_i
= (\tilde y_i + \alpha) - (x_i + \alpha)
= \tilde y_i - x_i,
\]
so the interval
\[
\bigl[\min_i(\tilde y_i-x_i),\,\max_i(\tilde y_i-x_i)\bigr]
\]
is the same for $(x,\tilde y)$ and $(x',\tilde y')$. In particular, a given
choice of $\ell$ in this interval is valid for both pairs of representatives.

Using max-plus notation, we compute
\[
\ell \odot x' \oplus \tilde y'
= \max(x'+\ell,\tilde y')
= \max(x+\alpha\ones+\ell,\tilde y+\alpha\ones)
= \max(x+\ell,\tilde y) + \alpha\ones
= \bigl(\ell\odot x \oplus \tilde y\bigr) + \alpha\ones.
\]
Thus the proposal $u'$ obtained from $(x',\tilde y')$ satisfies
$u' = u + \alpha\ones$, so $u$ and $u'$ define the same class in
$\TPS{\numpairs-1} = \R^{\numpairs}/\R\ones$. Therefore the point $u$ is
well defined on the quotient, and therefore it is independent of the choice of representatives for $x$ and $\tilde y$.
\end{proof}

\begin{remark}
\Cref{algo:trop_HAR_ultra} can be generalized to sample from any tropical polytope $\mathcal{P}$ if one uses \Cref{eq:trop_proj} to compute the tropical projection $\pi_\mathcal{P}(y)$ rather than the single-linkage algorithm. Recall, however, that \Cref{eq:trop_proj} requires us to have all the vertices of $\mathcal{P}$, which is not feasible for our problem.
\end{remark}

Using \cref{algo:trop_HAR_ultra}, we sample $N=1,000$ ultrametric from $\Uspace_6$, and record the topology\footnote{By ``topology'' of ultrametric $u$, we mean the index of the minimal closed cone of $\bergFan(K_6)$ containing $u$. Our \emph{acceptance} criteria in each case is that the sample has a fully-resolved topology, meaning that $u^{(k)}$ is contained in the relative interior of a maximal cone of $\bergFan(K_n)$. For $n=6$ (a small number of leaves), full enumeration of the generators (vertices) of $\bergFan(K_n)$ is manageable. As each cone is defined by a subset of these generators, checking containment can be performed as a linear programming feasibility problem. Other (faster) methods also exist.} of each collected tree metric. We sampled $y$ uniformly from the unit cube. For comparative purposes, we also include the ultrametrics obtained by excluding the ``run'' step, i.e., simply setting $u^{(k)}\leftarrow\tilde{y}$  in each iteration without sampling from the tropical line segment. Intuitively, we expect there to be a bias associated with ``no run'' sampling that is proportional to the \emph{size}\footnote{There are different notions of \emph{size} one can use. Here, we use the Euclidean volume of the cones' intersection with the unit hypercube. This is also why we sample $y$ from a hypercube as opposed to, say, the unit hypersphere - the intersections are easier to compute in our case.} of the topology's corresponding cone in $\tilde{\mathcal{B}}(K_n)$. A reasonable expectation might be that ``larger'' topologies (in the sense of their cones' size) will be sampled more frequently. We include a normalized measure of each cone's volume in our results and sort the topologies according to their relative sizes (largest to smallest). The results of this comparison are shown in \cref{fig:proj_HAR}.

\begin{figure}[h!]
    \centering
    \includegraphics[width=0.98\linewidth]{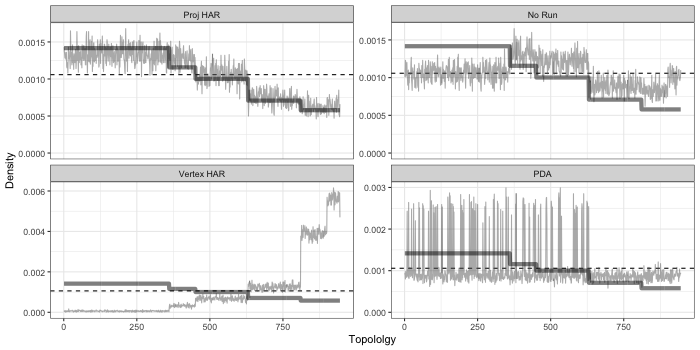}
    \caption{Density plots by tree topology for $N=1,000$ samples obtained using variations of the tropical HAR algorithm. The projective HAR method (top-left) utilizes \Cref{algo:trop_HAR_ultra}. The ``No Run'' method (top-right) also uses \Cref{algo:trop_HAR_ultra}, but skips lines 5-6 and sets $u^{(k)}\gets \tilde{y}$ directly. Vertex HAR (bottom-left) modifies \Cref{algo:trop_HAR_ultra} in line 3 by sampling $y$ from (the vertices of) $\mathcal{P}$ directly. PDA (bottom-right) modifies line 3 by computing $y$ using the \textsc{R} package \textsc{TreeTools}. The piecewise-linear lines represent the normalized cone volumes for each corresponding topology.}
    \label{fig:proj_HAR}
\end{figure}

It is interesting to note that HAR's sampling frequency is much more closely aligned with volume than that of the ``no run'' method, though both methods over- and under-sample considerably across the range of topologies. We also include the vertex HAR method, which selects $y$ from amongst the vertices of $\bergFan(K_6)$ at each step, and the PDA method from the \textsc{TreeTools}\footnote{By default, the \textsc{RandomTrees} function of \textsc{TreeTools} does not produce ultrametrics, however, such trees can be made into ultrametrics without changing the underlying topology by extending pendant edges as required.} package in \textsc{R}.

\section*{Conclusions}

The tropical Grassmannian offers a unifying, algebraically rigorous framework for understanding tree space. By tropicalizing the Plücker relations, we recover the four-point condition and ensure that every point of $\tropGr(2,n)$ corresponds to a valid tree metric. This reveals that the vast combinatorial complexity of phylogenetic trees can be encoded by relatively few parameters and expressed within the rich structure of tropical geometry. Moreover, the tropical viewpoint connects naturally to ultrametrics, Bergman fans, and novel distance measures such as the tropical metric, each of which provides fresh perspectives on phylogenetic inference.

At the same time, our exploration highlights practical challenges. Sampling from $\tropGr(2,n)$ via uniform draws in $\Gr(2,n)$ produces highly unbalanced trees, distinct from biologically realistic models such as Yule or coalescent processes. Likewise, while the tropical metric is computationally simple and respects combinatorial tree structure, its biological interpretation remains less clear compared to BHV, SPR, or RF distances.

Future work should focus on developing biologically meaningful sampling strategies in tropical space, connecting tropical geometry with stochastic models of evolution, and exploring the use of tropical optimization techniques in inference. In doing so, the tropical Grassmannian has the potential not only to deepen our theoretical understanding of tree space, but also to provide new algorithmic tools for phylogenetic reconstruction in large and complex datasets.

\printbibliography

@ARTICLE{Andersen2025-rb,
  title         = "Inhomogeneous branching trees with symmetric and asymmetric
                   offspring and their genealogies",
  author        = "Andersen, Frederik M and Suchard, Marc A and Wiuf, Carsten
                   and Bhatt, Samir",
  journal       = "arXiv [math.PR]",
  month         =  oct,
  year          =  2025,
  archivePrefix = "arXiv",
  primaryClass  = "math.PR"
}

@ARTICLE{Kendall1948-ht,
  title     = "On the Generalized ``Birth-and-Death'' Process",
  author    = "Kendall, David G",
  journal   = "aoms",
  publisher = "Institute of Mathematical Statistics",
  volume    =  19,
  number    =  1,
  pages     = "1--15",
  month     =  mar,
  year      =  1948,
  language  = "en"
}

@article{convexity,
author = {Lin, Bo and Sturmfels, Bernd and Tang, Xiaoxian and Yoshida, Ruriko},
title = {Convexity in Tree Spaces},
journal = {SIAM Journal on Discrete Mathematics},
volume = {31},
number = {3},
pages = {2015-2038},
year = {2017},
doi = {10.1137/16M1079841},

URL = { 
    
        https://doi.org/10.1137/16M1079841
    
    

},
eprint = { 
    
        https://doi.org/10.1137/16M1079841
    
    

}
}

@ARTICLE{Yule1925-de,
  title     = "{II}.—A mathematical theory of evolution, based on the
               conclusions of Dr. {J}. {C}. Willis, {F}. {R}. {S}",
  author    = "Yule, G U",
  journal   = "Philos. Trans. R. Soc. Lond.",
  publisher = "The Royal Society",
  volume    =  213,
  number    = "402-410",
  pages     = "21--87",
  month     =  jan,
  year      =  1925,
  language  = "en"
}

@ARTICLE{Minh2020-ey,
  title    = "{IQ}-{TREE} 2: New Models and Efficient Methods for Phylogenetic
              Inference in the Genomic Era",
  author   = "Minh, Bui Quang and Schmidt, Heiko A and Chernomor, Olga and
              Schrempf, Dominik and Woodhams, Michael D and von Haeseler, Arndt
              and Lanfear, Robert",
  journal  = "Mol. Biol. Evol.",
  volume   =  37,
  number   =  5,
  pages    = "1530--1534",
  month    =  may,
  year     =  2020,
  language = "en"
}

@ARTICLE{Robinson1981-nd,
  title   = "Comparison of phylogenetic trees",
  author  = "Robinson, D F and Foulds, L R",
  journal = "Math. Biosci.",
  volume  =  53,
  number  =  1,
  pages   = "131--147",
  month   =  feb,
  year    =  1981
}

@ARTICLE{Penn2024-ja,
  title     = "{Phylo2Vec}: a vector representation for binary trees",
  author    = "Penn, Matthew J and Scheidwasser, Neil and Khurana, Mark P and
               Duchêne, David A and Donnelly, Christl A and Bhatt, Samir",
  journal   = "Syst. Biol.",
  publisher = "Oxford University Press",
  pages     = "syae030",
  month     =  jun,
  year      =  2024,
  language  = "en"
}

@ARTICLE{Herrmann2009-ad,
  title     = "How to draw tropical planes",
  author    = "Herrmann, Sven and Jensen, Anders and Joswig, Michael and
               Sturmfels, Bernd",
  journal   = "Electron. J. Comb.",
  publisher = "The Electronic Journal of Combinatorics",
  volume    =  16,
  number    =  2,
  pages     = "R6",
  month     =  apr,
  year      =  2009
}

@ARTICLE{Dress1992-zm,
  title     = "Perfect matroids",
  author    = "Dress, Andreas W M and Wenzel, Walter",
  journal   = "Adv. Math. (N. Y.)",
  publisher = "Elsevier BV",
  volume    =  91,
  number    =  2,
  pages     = "158--208",
  month     =  feb,
  year      =  1992,
  language  = "en"
}

@ARTICLE{Owen2011-sm,
  title     = "A fast algorithm for computing geodesic distances in tree space",
  author    = "Owen, Megan and Provan, J Scott",
  journal   = "IEEE/ACM Trans. Comput. Biol. Bioinform.",
  publisher = "Institute of Electrical and Electronics Engineers (IEEE)",
  volume    =  8,
  number    =  1,
  pages     = "2--13",
  month     =  jan,
  year      =  2011,
  language  = "en"
}

@ARTICLE{Buneman1974-np,
  title     = "A note on the metric properties of trees",
  author    = "Buneman, Peter",
  journal   = "J. Combin. Theory Ser. B",
  publisher = "Elsevier BV",
  volume    =  17,
  number    =  1,
  pages     = "48--50",
  month     =  aug,
  year      =  1974,
  language  = "en"
}

@article{Yoshida2023,
title={Hit and run sampling from tropically convex sets},
author={Ruriko Yoshida and Keiji Miura and David Barnhill},
Volume =14, 
year=2023, 
Number= 1, 
pages={37--69},
}

@article{lin2017convexity,
  title={Convexity in tree spaces},
  author={Lin, Bo and Sturmfels, Bernd and Tang, Xiaoxian and Yoshida, Ruriko},
  journal={SIAM Journal on Discrete Mathematics},
  volume={31},
  number={3},
  pages={2015--2038},
  year={2017},
  publisher={SIAM}
}

@ARTICLE{Fitch1971-tz,
  title     = "Toward Defining the Course of Evolution: Minimum Change for a
               Specific Tree Topology",
  author    = "Fitch, Walter M",
  journal   = "Syst. Biol.",
  publisher = "Oxford Academic",
  volume    =  20,
  number    =  4,
  pages     = "406--416",
  month     =  dec,
  year      =  1971,
  language  = "en"
}

@ARTICLE{Roch2006-qs,
  title    = "A short proof that phylogenetic tree reconstruction by maximum
              likelihood is hard",
  author   = "Roch, Sebastien",
  journal  = "IEEE/ACM Trans. Comput. Biol. Bioinform.",
  volume   =  3,
  number   =  1,
  pages    = "92--94",
  year     =  2006,
  language = "en"
}

@ARTICLE{Foulds1982-wu,
  title     = "The steiner problem in phylogeny is {NP}-complete",
  author    = "Foulds, L R and Graham, R L",
  journal   = "Adv. Appl. Math.",
  publisher = "Elsevier BV",
  volume    =  3,
  number    =  1,
  pages     = "43--49",
  month     =  mar,
  year      =  1982,
  language  = "en"
}

@ARTICLE{Fiorini2012-lq,
  title     = "Approximating the balanced minimum evolution problem",
  author    = "Fiorini, Samuel and Joret, Gwenaël",
  journal   = "Oper. Res. Lett.",
  publisher = "Elsevier BV",
  volume    =  40,
  number    =  1,
  pages     = "31--35",
  month     =  jan,
  year      =  2012,
  language  = "en"
}

@ARTICLE{Felsenstein1983-wd,
  title     = "Statistical inference of phylogenies",
  author    = "Felsenstein, Joseph",
  journal   = "J. R. Stat. Soc. Ser. A",
  publisher = "JSTOR",
  volume    =  146,
  number    =  3,
  pages     =  246,
  year      =  1983
}

@ARTICLE{Rzhetsky1993-vo,
  title    = "Theoretical foundation of the minimum-evolution method of
              phylogenetic inference",
  author   = "Rzhetsky, A and Nei, M",
  journal  = "Mol. Biol. Evol.",
  volume   =  10,
  number   =  5,
  pages    = "1073--1095",
  month    =  sep,
  year     =  1993,
  language = "en"
}

@ARTICLE{Billera2001-il,
  title    = "Geometry of the Space of Phylogenetic Trees",
  author   = "Billera, Louis J and Holmes, Susan P and Vogtmann, Karen",
  journal  = "Adv. Appl. Math.",
  volume   =  27,
  number   =  4,
  pages    = "733--767",
  month    =  nov,
  year     =  2001
}

@ARTICLE{Speyer2004-bg,
  title     = "The tropical grassmannian",
  author    = "Speyer, David and Sturmfels, Bernd",
  journal   = "Adv. Geom.",
  publisher = "Walter de Gruyter GmbH",
  volume    =  4,
  number    =  3,
  pages     = "389--411",
  month     =  jul,
  year      =  2004,
  language  = "en"
}

@ARTICLE{Gascuel2006-sr,
  title    = "Neighbor-joining revealed",
  author   = "Gascuel, Olivier and Steel, Mike",
  journal  = "Mol. Biol. Evol.",
  volume   =  23,
  number   =  11,
  pages    = "1997--2000",
  month    =  nov,
  year     =  2006,
  language = "en"
}

@article{Ardilasubdominant,
  title={Subdominant matroid ultrametrics},
  author={Ardila, Federico},
  journal={Annals of Combinatorics},
  volume={8},
  number={4},
  pages={379--389},
  year={2005},
  publisher={Springer}
}

@book{maclagan2015introduction,
  title={Introduction to Tropical Geometry},
  author={Maclagan, D. and Sturmfels, B.},
  isbn={9780821851982},
  lccn={2014036141},
  series={Graduate Studies in Mathematics},
  url={https://books.google.com/books?id=3DLMoQEACAAJ},
  year={2015},
  publisher={American Mathematical Society}
}

@article{Fink_2015_Steifel,
author = {Alex Fink and Felipe Rincón},
title = {Stiefel tropical linear spaces},
journal = {Journal of Combinatorial Theory, Series A},
volume = {135},
pages = {291-331},
year = {2015},
issn = {0097-3165},
doi = {https://doi.org/10.1016/j.jcta.2015.06.001},
url = {https://www.sciencedirect.com/science/article/pii/S0097316515000710},
}

@book{gordon2012matroids,
  title={Matroids: A Geometric Introduction},
  author={Gordon, G. and McNulty, J.},
  isbn={9780521145688},
  lccn={2012009109},
  url={https://books.google.com/books?id=vgC_4B1AhGQC},
  year={2012},
  publisher={Cambridge University Press}
}

@book{oxley2006matroid,
  title={Matroid Theory},
  author={Oxley, J.G.},
  isbn={9780199202508},
  lccn={92020802},
  series={Oxford graduate texts in mathematics},
  url={https://books.google.com/books?id=puKta1Hdz-8C},
  year={2006},
  publisher={Oxford University Press}
}

@article{speyer2008tropLinearSpaces,
    author = {Speyer, David E.},
    title = {Tropical Linear Spaces},
    journal = {SIAM Journal on Discrete Mathematics},
    volume = {22},
    number = {4},
    pages = {1527-1558},
    year = {2008},
    doi = {10.1137/080716219},
    URL = {https://doi.org/10.1137/080716219},
    eprint = {https://doi.org/10.1137/080716219}
}
\newpage

\end{document}